\documentclass[letterpaper,10 pt,conference]{ieeeconf}
\IEEEoverridecommandlockouts 
\usepackage{amsmath,amssymb,amsfonts}
\usepackage{graphicx}
\usepackage{subfigure}
\usepackage{textcomp}
\usepackage{color}
\usepackage{xcolor}
\usepackage{multirow}
\usepackage{multicol,lipsum}
\usepackage{epsfig}
\usepackage{algorithm}
\usepackage{algorithmic}
\usepackage{multirow}
\usepackage{hyperref}
\usepackage{mathtools}
\usepackage{bm}
\usepackage{listings}
\newtheorem{remark}{Remark}
\newtheorem{definition}{Definition}
\newtheorem{theorem}{Theorem}
\newtheorem{proposition}{Proposition}
\newtheorem{lemma}{Lemma}
\newtheorem{example}{Example}
\newtheorem{corollary}{Corollary}

\usepackage{array}
\usepackage[short]{optidef}
\usepackage{mdframed}
\mdfdefinestyle{mystyle}{
    backgroundcolor=white!10}

\newcommand{\oomit}[1]{}

\newcolumntype{M}[1]{>{\centering\arraybackslash}m{#1}}
\newcolumntype{N}{@{}m{0pt}@{}}

\begin{document}

\title{Reach-avoid Verification using Lyapunov Densities}

\author{Bai Xue
\thanks{State Key Lab. of Computer Science, Institute of Software, CAS, Beijing, China \{xuebai@ios.ac.cn\} 
}
\thanks{University of Chinese Academy of Sciences, Beijing, China}
}

\maketitle
\thispagestyle{empty}
\pagestyle{empty}

\begin{abstract}
Reach-avoid analysis combines the construction of safety and specific progress guarantees, and is able to formalize many important engineering problems. In this paper we study the reach-avoid verification problem of systems modelled by ordinary differential equations using Lyapunov densities. Firstly, the weak reach-avoid verification is considered. Given an initial set, a safe set and a target set, the weak reach-avoid verification is to verify whether the reach-avoid property (i.e., the system will enter the target set eventually while staying inside the safe set before the first target hitting time) holds for almost all states in the initial set. We propose two novel sufficient conditions using Lyapunov densities for the weak reach-avoid verification. These two sufficient conditions are shown to be weaker than existing ones, providing more possibilities of verifying weak reach-avoid properties successfully. Then, we generalize these conditions to the verification of reach-avoid properties for all states in the initial set. Finally, an example demonstrates theoretical developments of proposed conditions. 
\end{abstract}
\section{Introduction}
\label{sec:intro}
Reach-avoid analysis combines the construction of safety and specific progress guarantees for dynamical systems, as it addresses guarantees for both the eventual reach of desirable states and avoidance of unsafe states. It is employed in diverse engineering applications such as motion planning \cite{herbert2017fastrack}. Reach-avoid analysis in this paper mainly attempts to verify reach-avoid properties, i.e., verify whether a system starting from a legally initial set will enter a desirable target set eventually while reliably avoiding a set of unsafe states before hitting the target set. 

Various methods have been applied to certify the reach-avoid properties of engineering systems, e.g., \cite{margellos2011,korda2014controller,xue2016reach}. One of well-known methods is the barrier certificate method, which was originally proposed for safety verification of dynamical systems in \cite{PrajnaJP07} and then extended to reach-avoid verification in \cite{prajna2007convex}. Recently, guidance-barrier functions were  proposed in \cite{xue2022reach} for reach-avoid verification. These methods investigate reach-avoid properties of nonlinear dynamical systems without explicitly computing the solutions of these systems, as done in the stability analysis with Lyapunov functions. However, one may not be able to find a function to certify the reach-avoid property due to the fact that the solution trajectory for some initial state, which is in a negligible set (i.e., a set with measure zero), may not reach a desired set. Therefore, Lyapunov densities have been used to verify weak reach-avoid properties of nonlinear systems in \cite{prajna2007convex}. The notion `weak' is used to emphasize that the system satisfies a property for almost all points in the domain. Specially, Lyapunov densities evaluate how the measure of a set is evolving along the solutions. Thus, they provide certifications for not all points but almost every point in the domain. 

In this paper we investigate the problem of reach-avoid verification for systems modeled by ordinary differential equations. The main results presented in this paper for reach-avoid verification rely on Lyapunov densities. Firstly, the weak reach-avoid verification is considered. Given an initial set, a safe set and a target set, the weak reach-avoid verification is to verify the satisfaction of the weak reach-avoid property, which formulates that the system starting from almost all states in the initial set will enter the target set eventually while staying inside the safe set before the first target hitting time. Inspired by the conditions proposed in \cite{xue2022reach} for the strong reach-avoid verification (i.e., verify the satisfaction of the reach-avoid property for all states in the initial set), we propose two sufficient conditions in the density space for verifying the weak reach-avoid property.  Then, via analyzing the divergence of the vector field of the system, we establish the relationship between these two conditions and the ones in \cite{xue2022reach}, and further generalize these two conditions to the strong reach-avoid verification. \oomit{Afterwards, we construct a convex optimization, based on encoding one of these two conditions into semi-definite constraints via sum-of-squares decomposition for multivariate polynomials, for synthesizing weak reach-avoid controllers. }Finally, we demonstrate the theoretical developments of proposed methods on one example.

The contributions of this work are summarized below.
\begin{enumerate}
    \item Two novel conditions in the density space are proposed for the weak reach-avoid verification of systems modelled by ordinary differential equations. These two conditions are shown to be weaker than the one in \cite{prajna2007convex}, providing more possibilities of verifying weak reach-avoid properties successfully. 
    \item We generalize the conditions for the weak reach-avoid verification to the strong one, lifting their capabilities in reach-avoid verification.
\end{enumerate}

\subsection*{Related Work}
There are a large amount of works on reach-avoid analysis, e.g., \cite{summers2010verification,fisac2015reach,xue2016reach,ding2010robust,kochdumper2021aroc,fan2018controller,xue2019,margellos2011hamilton}. Thus, we do not intend to provide a comprehensive and thorough literature review, but rather present some closely related works here.

Lyapunov density has been first introduced in \cite{RANTZER2001161} as a tool to certify almost global stability of nonlinear systems. Almost global stability of origins means that the solutions converge to the origin for almost every initial state. In \cite{karabacak2019almost}, the result on nonlinear systems obtained in \cite{RANTZER2001161} has been generalized to nonlinear systems with time dependent switching. Sufficient conditions to ensure almost global stability of nonlinear systems with time dependent switching have been provided with the help of common Lyapunov density and multiple Lyapunov densities. Moreover, Lyapunov densities were used to certify almost global stability of nonautonomous systems \cite{monzon2005almost,masubuchi2017lyapunov} and switched systems with state-dependent switching \cite{masubuchi2016analysis}. Recently, they have been extended to the verification of temporal properties of nonlinear systems such as safety and reach-avoidance for nonlinear systems. Some sufficient conditions has been developed for nonlinear (disturbed) systems. Leaning upon the results in \cite{prajna2007convex}, certificates for weak safety and weak reach-avoid verification of nonlinear (switched) systems with time dependent switching have been given based on Lyapunov densities in \cite{kivilcim2019safety} and \cite{kivilcim2019safe}. Afterwards, the result in \cite{kivilcim2019safe} was extended to the weak reach-avoid verification of nonlinear systems with state-dependent switching in \cite{kivilcim2021weak}. In this paper we proposed two new sufficient conditions based on Lyapunov densities to the weak reach-avoid verification of nonlinear systems. They are shown to be weaker than the one in \cite{prajna2007convex}. Furthermore, we extend them to the strong reach-avoid verification.

\section{Preliminaries}
\label{sec:preliminaries}
We denote the space of $m$-times continuously differentiable functions mapping $\mathcal{X}\subseteq \mathbb{R}^n$ to $\mathbb{R}^p$ by $\mathcal{C}^m(\mathcal{X},\mathbb{R}^p)$. When $p=1$, we will simply write $\mathcal{C}^m(\mathcal{X})$, and for continuous functions ($m=0$), we will omit the superscript.  For a function $\bm{f}(\cdot):\mathbb{R}^n \rightarrow \mathbb{R}^n$ with $\bm{f}=(f_1,\ldots,f_n)^{\top}$, $\triangledown \bm{f}=\sum_{i=1}^n \frac{\partial f_i}{\partial x_i}$ denotes the divergence of $\bm{f}$; for a function $g(\cdot): \mathbb{R}^n \rightarrow \mathbb{R}$, $\triangledown g=(\frac{\partial g}{\partial x_1}, \ldots, \frac{\partial g}{\partial x_n})$ denotes the gradient of $g$; given a set $\mathcal{X}$, the sets $\partial \mathcal{X}$ and $\overline{\mathcal{X}}$ denote its boundary and closure, respectively.

\subsection{Problem Statement}
In this subsection we formulate the system and its associated strong/weak reach-avoid properties of interest.

The system of interest is a system whose dynamics are described by an ODE of the following form:
\begin{equation}
\label{sys}
\dot{\bm{x}}(t)=\bm{f}(\bm{x}(t)),
\bm{x}(0)=\bm{x}_0\in \mathbb{R}^n,
\end{equation}
where $\dot{\bm{x}}(t)=\frac{d \bm{x}(t)}{d t}$ and $\bm{f}(\bm{x})=(f_1(\bm{x}),\ldots,f_n(\bm{x}))^{\top}$ with $f_i(\bm{x})$ being locally Lipschitz continuous.

We denote the trajectory of system \eqref{sys} that originates from $\bm{x}_0 \in \mathbb{R}^n$ and is defined over the maximal time interval $[0,T_{\bm{x_0}})$ by $\bm{\phi}_{\bm{x}_0}(\cdot):[0,T_{\bm{x_0}})\rightarrow \mathbb{R}^n$. Consequently,
\[\bm{\phi}_{\bm{x}_0}(t):=\bm{x}(t), \forall t\in [0,T_{\bm{x_0}}), \text{ and } \bm{\phi}_{\bm{x}_0}(0) = \bm{x}_0,\] where $T_{\bm{x}_0}$ is either a positive value or $\infty$.

Given a \textit{bounded and open} safe set $\mathcal{X}$, an initial set $\mathcal{X}_0$ and a compact target set $\mathcal{X}_r$, where 
\begin{equation*}
\label{set}
\begin{split}
&\mathcal{X}=\{\bm{x}\in \mathbb{R}^n\mid h(\bm{x})<0\} \text{~with~} \partial \mathcal{X}=\{\bm{x}\in \mathbb{R}^n\mid h(\bm{x})=0\},\\
&\mathcal{X}_0=\{\bm{x}\in \mathbb{R}^n\mid l(\bm{x})< 0\}, \text{~and~}\\
&\mathcal{X}_r=\{\bm{x}\in \mathbb{R}^n\mid g(\bm{x})\leq 0\}
\end{split}
\end{equation*}
with $l(\bm{x}), h(\bm{x}),g (\bm{x}): \mathbb{R}^n \rightarrow \mathbb{R}$, and $\mathcal{X}_0, \mathcal{X}_r\subseteq \mathcal{X}$, both strong and weak reach-avoid properties are defined below.
\begin{definition}[Strong Reach-avoid Property]\label{SRNS}
Given system \eqref{sys} with the safe set $\mathcal{X}$, initial set $\mathcal{X}_0$ and target set $\mathcal{X}_r$, we say that the \textit{strong} reach-avoid property holds if for any initial condition $\bm{x}_0\in \mathcal{X}_0$, its trajectory $\bm{\phi}_{\bm{x}_0}(t)$ satisfies \[\bm{\phi}_{\bm{x}_0}(T)\in \mathcal{X}_r \bigwedge \forall t\in [0,T]. \bm{\phi}_{\bm{x}_0}(t)\in \mathcal{X}\] for some $T>0$.
\end{definition}

\begin{definition}[Weak Reach-avoid Property]\label{WRNS}
Given system \eqref{sys} with the safe set $\mathcal{X}$, initial set $\mathcal{X}_0$ and target set $\mathcal{X}_r$, we say that the \textit{weak} reach-avoid property holds if for \textit{almost all} initial conditions $\bm{x}_0\in \mathcal{X}_0$, its trajectory $\bm{\phi}_{\bm{x}_0}(t)$ satisfies \[\bm{\phi}_{\bm{x}_0}(T)\in \mathcal{X}_r \bigwedge \forall t\in [0,T]. \bm{\phi}_{\bm{x}_0}(t)\in \mathcal{X}\] for some $T>0$.
\end{definition}

\subsection{Conditions for Reach-avoid Verification}
In this subsection we recall existing sufficient conditions for assuring the satisfaction of strong and weak reach-avoid properties.

\begin{proposition}[Proposition 5, \cite{xue2022reach}]
\label{inequality1}
 Given system \eqref{sys} with sets $\mathcal{X}_0$, $\mathcal{X}_r$ and $\mathcal{X}$, if there exists a continuously differentiable function $v(\bm{x})\in \mathcal{C}^1(\overline{\mathcal{X}})$ such that 
 \begin{equation}
 \label{strong_exponen_con}
\begin{cases}
&v(\bm{x})>0, \forall \bm{x}\in \mathcal{X}_0\\
&\bigtriangledown  v(\bm{x}) \cdot \bm{f}(\bm{x})\geq \lambda v (\bm{x}), \forall \bm{x}\in \overline{\mathcal{X}\setminus \mathcal{X}_r} \\
&v(\bm{x})\leq 0, \forall \bm{x}\in \partial \mathcal{X}
\end{cases}
\end{equation}
where $\lambda>0$ is a user-defined value, then the strong reach-avoid property in Definition \ref{SRNS} holds.
\end{proposition}

\begin{proposition}[Proposition 4, \cite{xue2022reach}]
\label{inequality}
  Given system \eqref{sys} with sets $\mathcal{X}_0$, $\mathcal{X}_r$ and $\mathcal{X}$, if there exist a  continuously differentiable function $v(\bm{x})\in \mathcal{C}^1(\overline{\mathcal{X}})$ and a continuously differentiable function $w(\bm{x})\in \mathcal{C}^1(\overline{\mathcal{X}})$ satisfying 
 \begin{equation}
 \label{strong_asympt_con}
\begin{cases}
 v(\bm{x})>0, \forall \bm{x}\in \mathcal{X}_0\\
\bigtriangledown  v(\bm{x}) \cdot \bm{f}(\bm{x})\geq 0, \forall \bm{x}\in \overline{\mathcal{X}\setminus \mathcal{X}_r}\\
v(\bm{x})-\bigtriangledown  w(\bm{x}) \cdot \bm{f}(\bm{x})\leq 0, \forall \bm{x}\in \overline{\mathcal{X}\setminus \mathcal{X}_r} \\
v(\bm{x})\leq 0, \forall \bm{x}\in \partial \mathcal{X}
\end{cases}
\end{equation}
then the strong reach-avoid property in Definition \ref{SRNS} holds.
\end{proposition}

The condition for the weak reach-avoid verification is presented in \cite{prajna2007convex}.

\begin{proposition}[Corollary 3.8, \cite{prajna2007convex}]
\label{weak_density}
    Given system \eqref{sys} with sets $\mathcal{X}_0$, $\mathcal{X}_r$ and $\mathcal{X}$, and an open set $\hat{\mathcal{X}}$ containing $\mathcal{X}_0$, if there exists a density function $\rho(\bm{x})\in \mathcal{C}^1(\overline{\mathcal{X}})$ which satisfies
\begin{equation}
\label{weak_2}
    \begin{cases}
        \rho(\bm{x})>0, \forall \bm{x}\in \hat{\mathcal{X}},\\
        \rho(\bm{x})\leq 0, \forall \bm{x}\in \partial \mathcal{X},\\
        \triangledown(\rho\bm{f})(\bm{x})>0, \forall \bm{x}\in \overline{\mathcal{X}\setminus \mathcal{X}_r},
    \end{cases}
\end{equation}
 where $\triangledown (\rho \bm{f})(\bm{x})=\triangledown \big(\rho(\bm{x}) \cdot \bm{f}(\bm{x})\big)=\triangledown \rho(\bm{x})\bm{f}(\bm{x})+\rho(\bm{x})\triangledown \bm{f}(\bm{x})$, then the weak reach-avoid property in Definition \ref{WRNS} holds.
\end{proposition}

In \cite{prajna2007convex}, the constraint $\rho(\bm{x})\leq 0, \forall \bm{x}\in \overline{\partial \mathcal{X}\setminus \partial \mathcal{X}_r}$ rather than $\rho(\bm{x})\leq 0, \forall \bm{x}\in \partial \mathcal{X}$ is used, since $\partial \mathcal{X}\setminus \partial \mathcal{X}_r=\partial \mathcal{X}$ in this paper (it can be justified according to the fact that $\mathcal{X}$ is open, $\mathcal{X}_r$ is compact and $\mathcal{X}_r\subseteq \mathcal{X}$).

An obvious deficiency of the condition in Proposition \ref{weak_density} is that it is not applicable to system \eqref{sys} with 
\begin{equation}
\label{equaltozerro}
\exists \bm{x}_0 \in \mathcal{X}\setminus \mathcal{X}_r. \bm{f}(\bm{x}_0)=0 \wedge \sum_{i=1}^n \frac{\partial f_i(\bm{x}_0)}{\partial x_i}=0, 
\end{equation}
which results in $\triangledown (\rho\bm{f})(\bm{x}_0)=0$ for any $\rho(\bm{x})\in \mathcal{C}^1(\overline{\mathcal{X}})$.
However, the sufficient conditions proposed in the present work will apply to this system. Moreover, they are more expressive than condition \eqref{weak_2}.

\section{Reach-avoid Verification}
\label{sec:synthesis}
In this section we present our sufficient conditions for verifying the weak reach-avoid property in Definition \ref{WRNS}. These sufficient conditions are inspired by those in Proposition \ref{inequality1} and \ref{inequality} as well as density functions in \cite{prajna2007convex}. Afterwards, we exploit the relationship between the derived conditions and those in Proposition  \ref{inequality1} and \ref{inequality}, and formulate the situation under which the derived conditions can also be used to verifying the strong reach-avoid property in Definition \ref{SRNS}.

\subsection{Weak Reach-avoid Verification}
In this subsection we present our sufficient conditions for verifying the weak reach-avoid property in Definition \ref{WRNS}. The derivation of these conditions partly relies on Liouville's theorem \cite{RANTZER2001161}, which is formulated in Lemma \ref{density}. 

\begin{lemma}
\label{density}
    Let $\bm{f}=(f_1,\ldots,f_n)^{\top}\in \mathcal{C}^1(\mathcal{D},\mathbb{R}^n)$, where $\mathcal{D}\subseteq \mathbb{R}^n$ is open, and $\rho\in \mathcal{C}^1(\mathcal{D})$ be integrable. For $\bm{x}_0\in \mathbb{R}^n$, let $\bm{\phi}_{\bm{x}_0}(t)$ be the solution to system \eqref{sys} with $\bm{x}(0)=\bm{x}_0$. For a measurable set $Z$, assume that $\bm{\phi}_{Z}(\tau)=\{\bm{\phi}_{\bm{x}_0}(\tau)\mid \bm{x}_0\in Z\}$ is a subset of $\mathcal{D}$ for all $\tau \in [0,t]$. Then 
    \[\int_{\bm{\phi}_{Z}(t)}\rho(\bm{x})d\bm{x}-\int_{Z}\rho(\bm{x})d\bm{x}=\int_{0}^t \int_{\bm{\phi}_{Z}(\tau)}\triangledown (\rho \bm{f})(\bm{x})d\bm{x}d\tau.\]
\end{lemma}

Our first sufficient condition, which is adapted from the one in Proposition \ref{inequality1}, for verifying the weak reach-avoid property in Definition \ref{WRNS} is formulated in Theorem \ref{exponential_almost}. 
\begin{theorem}
\label{exponential_almost}
Consider system \eqref{sys} with the safe set $\mathcal{X}$, target set $\mathcal{X}_r$ and initial set $\mathcal{X}_0$. Given a continuous function $\lambda(\bm{x})>0$ over $\overline{\mathcal{X}\setminus \mathcal{X}_r}$, if there exists a density function $\rho(\bm{x})\in \mathcal{C}^1(\overline{\mathcal{X}})$ satisfying 
 \begin{equation}
 \label{weak_exponential_con}
 \begin{cases}
 \rho(\bm{x})>0, \forall \bm{x}\in \mathcal{X}_0,\\
 \triangledown (\rho \bm{f})(\bm{x})\geq \lambda(\bm{x}) \rho(\bm{x}), \forall \bm{x}\in \overline{\mathcal{X}\setminus \mathcal{X}_r},\\
 \rho(\bm{x})\leq 0, \forall \bm{x}\in \partial \mathcal{X},
 \end{cases}
 \end{equation}
 then the weak reach-avoid property in Definition \ref{WRNS} holds.
\end{theorem}
\begin{proof}
Since $\mathcal{X}$ is bounded, $ \overline{\mathcal{X}\setminus \mathcal{X}_r}$ is compact. Therefore, there exists $\delta>0$ such that 
\[\lambda(\bm{x})\geq \delta, \forall \bm{x}\in  \overline{\mathcal{X}\setminus \mathcal{X}_r}.\]

We first show that given $\bm{x}_0\in \mathcal{R}=\{\bm{x}\in \mathcal{X}\mid  \rho(\bm{x})>0\}$, if system \eqref{sys} leaves $\mathcal{R}$, it must enter $\mathcal{X}_r$ before leaving  $\mathcal{R}$. Suppose to the contrary that the flow $\bm{\phi}_{\bm{x}_0}(t)$ leaves $\mathcal{R}$ without entering $\mathcal{X}_r$ first. Let $T>0$ be the first time instant that $\bm{\phi}_{\bm{x}_0}(t)$ leaves $\mathcal{R}$. By this we mean that $\bm{\phi}_{\bm{x}_0}(t) \in \mathcal{R}\setminus \mathcal{X}_r$ for all $t\in [0,T)$ and $\bm{\phi}_{\bm{x}_0}(T)\in \partial \mathcal{R}$ (i.e., $\rho(\bm{\phi}_{\bm{x}_0}(T))=0$ and $\rho(\bm{\phi}_{\bm{x}_0}(t))>0$ for $t\in [0,T)$). Also, since 
\[\triangledown(\rho\bm{f})(\bm{x})\geq \lambda(\bm{x}) \rho(\bm{x}), \forall \bm{x}\in \overline{\mathcal{X}\setminus \mathcal{X}_r},\]
we have that 
\[\triangledown(\rho\bm{f})(\bm{x})\mid_{\bm{x}=\bm{\phi}_{\bm{x}_0}(t)}\geq \lambda(\bm{\phi}_{\bm{x}_0}(t))\rho(\bm{\phi}_{\bm{x}_0}(t)), \forall t\in [0,T].\]That is, 
\[
\begin{split}
\frac{d \rho(\bm{\phi}_{\bm{x}_0}(t))}{d t}&=\triangledown \rho(\bm{x})\cdot \bm{f}(\bm{x})\mid_{\bm{x}=\bm{\phi}_{\bm{x}_0}(t)}\\
&\geq (\lambda(\bm{\phi}_{\bm{x}_0}(t))-\triangledown \bm{f}(\bm{\phi}_{\bm{x}_0}(t)))\rho(\bm{\phi}_{\bm{x}_0}(t))\\
&\geq (\lambda(\bm{\phi}_{\bm{x}_0}(t))-\lambda_0) \rho(\bm{\phi}_{\bm{x}_0}(t))\\
&\geq (\delta-\lambda_0) \rho(\bm{\phi}_{\bm{x}_0}(t)), \forall t\in [0,T].
\end{split}
\]
where $\lambda_0=\max_{\bm{x}\in \overline{\mathcal{X}\setminus \mathcal{X}_r}} \triangledown \bm{f}(\bm{x})$. Thus, $\frac{d \big(-\rho(\bm{\phi}_{\bm{x}_0}(t))\big)}{d t} \leq (\delta-\lambda_0) (-\rho(\bm{\phi}_{\bm{x}_0}(t))), \forall t\in [0,T]$.  According to the Grönwall's inequality, we have that $\rho(\bm{\phi}_{\bm{x}_0}(T))\geq e^{\delta-\lambda_0} \rho(\bm{x}_0)$. This implies that $\rho(\bm{\phi}_{\bm{x}_0}(T))>0$, which contradicts $\rho(\bm{\phi}_{\bm{x}_0}(T))=0$. Therefore, there does not exist a trajectory which, starting from $\mathcal{R}$, will leave the set $\mathcal{R}$  before entering the target set $\mathcal{X}_r$.

Next, we show that the set of all initial conditions $\bm{x}_0$'s in $\mathcal{R}$ whose flows $\bm{\phi}_{\bm{x}_0}(t)$'s do not leave $\mathcal{R}\setminus \mathcal{X}_r$ in finite time is a set of measure zero. For these trajectories, $\rho(\bm{\phi}_{\bm{x}_0}(t))>0$ for $t\geq 0$. Now define
\begin{equation}
\label{Z}
Z=\bigcap_{i=1,2,\ldots}\{\bm{x}_0\in \mathcal{R}\mid \bm{\phi}_{\bm{x}_0}(t)\in \mathcal{R}\setminus \mathcal{X}_r, \forall t\in [0,i]\}.
\end{equation}
The set $Z$ is an intersection of countable open sets and hence is measurable. It contains all initial states in $\mathcal{R}$ for which the trajectories stay in $\mathcal{R}\setminus \mathcal{X}_r$ for all $t\geq 0$. That $Z$ is a set of measure zero can be shown using Lemma \ref{density} as follows. We have that 
\[\begin{split}
&\int_{\bm{\phi}_{Z}(t)}\rho(\bm{x}) d\bm{x}-\int_Z \rho(\bm{x}) d\bm{x}=\int_{0}^t \int_{\bm{\phi}_{Z}(\tau)}\triangledown (\rho \bm{f})(\bm{x})d\bm{x}d\tau\\
&\geq  \int_{0}^t \int_{\bm{\phi}_{Z}(\tau)} \lambda(\bm{x})\rho(\bm{x}) d\bm{x} d\tau\\
&\geq \delta \int_{0}^t \int_{\bm{\phi}_{Z}(\tau)} \rho(\bm{x}) d\bm{x} d\tau, \forall t\geq 0.
\end{split}
\]
where $\bm{\phi}_{Z}(t)=\{\bm{x}\mid \bm{x}=\bm{\phi}_{\bm{x}_0}(t),\bm{x}_0\in Z\}$. Let $\psi(t)=\int_{\bm{\phi}_{Z}(t)}\rho(\bm{x}) d\bm{x}$. Thus, 
$-\psi(t)\leq -\psi(0)+\delta \int_{0}^t (-\psi(\tau)) d\tau, \forall t\geq 0$, according to the Grönwall's inequality (integral form), we have $-\psi(t)\leq -e^{\delta t} \psi(0)$ for $t\geq 0$. Since $\rho(\bm{x})$ is bounded over $\mathcal{X}$, the measure of $Z$ is zero. 

Since $\rho(\bm{x})>0$ for $\bm{x}\in \mathcal{X}_0$, $\mathcal{X}_0\subseteq \mathcal{R}$ holds. Consequently, the conclusion holds.  
\end{proof}

From the proof of Theorem \ref{exponential_almost}, we observe that  if there exists a density function $\rho(\bm{x})\in \mathcal{C}^1(\overline{\mathcal{X}})$ satisfying condition \eqref{weak_exponential_con}, system \eqref{sys} starting from an initial state $\bm{x}_0\in \mathcal{R}$ will either stay inside $\mathcal{R}\setminus \mathcal{X}_r$ for all the time or enter the target set $\mathcal{X}_r$ in finite time while staying inside the safe set $\mathcal{R}\setminus \mathcal{X}_r$ before the first target hitting time. Moreover, the measure of initial states in $\mathcal{R}$ such that system \eqref{sys} starting from them will stay inside $\mathcal{R}\setminus \mathcal{X}_r$ for all the time is zero.

Comparing condition \eqref{weak_exponential_con} with condition \eqref{weak_2}, we observe that the term $\triangledown(\rho\bm{f})(\bm{x})$ in condition \eqref{weak_exponential_con} is required to be larger than zero only in the subset $\mathcal{R}\setminus \mathcal{X}_r$ rather than $\overline{\mathcal{X}\setminus \mathcal{X}_r}$. It can be non-positive in $\overline{\mathcal{X}\setminus(\mathcal{R}\setminus \mathcal{X}_r)}$. This renders condition \eqref{weak_exponential_con} applicable to the weak reach-avoid verification of system \eqref{sys} subject to \eqref{equaltozerro}. Besides, we can also conclude that if $\rho(\bm{x})\in \mathcal{C}^1(\overline{\mathcal{X}})$ satisfies $\triangledown (\rho \bm{f})(\bm{x})>0, \forall \bm{x}\in \overline{\mathcal{X}\setminus \mathcal{X}_r}$,  there exists $\lambda>0$ such that it satisfies $\triangledown (\rho \bm{f})(\bm{x})\geq \lambda \rho(\bm{x}), \forall \bm{x}\in \overline{\mathcal{X}\setminus \mathcal{X}_r}$.  This conclusion can be certified in the following way: That \[\triangledown (\rho \bm{f})(\bm{x})>0, \forall \bm{x}\in \overline{\mathcal{X}\setminus \mathcal{X}_r}\] implies that
\[\exists \epsilon_0>0. \triangledown (\rho \bm{f})(\bm{x})\geq \epsilon_0, \forall \bm{x}\in \overline{\mathcal{X}\setminus \mathcal{X}_r}.\]  Let $M\epsilon_0 \geq \rho(\bm{x})$ for $\bm{x}\in \overline{\mathcal{X}\setminus \mathcal{X}_r}$, where $M>0$. Therefore, we have \[\triangledown (\rho \bm{f})(\bm{x})\geq \lambda \rho(\bm{x}),  \forall \bm{x}\in \overline{\mathcal{X}\setminus \mathcal{X}_r},\]
where $\lambda=\frac{1}{M}$. Thus, condition \eqref{weak_exponential_con} is more expressive than \eqref{weak_2}.

It is worth noting that $\lambda(\bm{x})$ in condition \eqref{weak_exponential_con} should be strictly larger than zero over $ \overline{\mathcal{X}\setminus \mathcal{X}_r}$. If $\lambda(\bm{x}_0)=0$ for some $\bm{x}_0\in  \overline{\mathcal{X}\setminus \mathcal{X}_r}$, inspired by condition \eqref{strong_asympt_con}, we will present another condition for the weak reach-avoid verification.  
\begin{theorem}
\label{sdf}
Consider system \eqref{sys} with the safe set $\mathcal{X}$, target set $\mathcal{X}_r$ and initial set $\mathcal{X}_0$. Given a continuous function $\lambda(\bm{x})\geq 0$ over $\overline{\mathcal{X}\setminus \mathcal{X}_r}$, if there exist  density functions $\rho_1(\bm{x}), \rho_2(\bm{x})\in \mathcal{C}^1(\overline{\mathcal{X}})$ satisfying 
 \begin{equation}
 \label{weak_asympt_con}
 \begin{cases}
 \rho_1(\bm{x})>0, \forall \bm{x} \in \mathcal{X}_0,\\
 \triangledown (\rho_1 \bm{f})(\bm{x})\geq \lambda(\bm{x}) \rho_1(\bm{x}), \forall \bm{x}\in \overline{\mathcal{X}\setminus \mathcal{X}_r},\\
     \rho_1(\bm{x})\leq \triangledown (\rho_2\bm{f})(\bm{x}), \forall \bm{x}\in \overline{\mathcal{X}\setminus \mathcal{X}_r},\\
 \rho_1(\bm{x})\leq 0, \forall \bm{x}\in \partial\mathcal{X},
 \end{cases}
 \end{equation}
 then the weak reach-avoid property in Definition \ref{WRNS} holds. 
\end{theorem}
\begin{proof}
We first prove that the set of all initial states $\bm{x}_0$'s in $\mathcal{R}=\{\bm{x}\in \mathcal{X}\mid \rho(\bm{x})>0\}$ whose flows $\bm{\phi}_{\bm{x}_0}(t)$'s  do not leave the open set $\mathcal{R}\setminus \mathcal{X}_r$ in finite time is a set of measure zero. We show that the measure of the set $Z$ in \eqref{Z} is zero. Since $\bm{\phi}_{Z}(t)\subseteq  \mathcal{R}\setminus \mathcal{X}_r$ for $t\geq 0$, $\mathcal{R}\setminus \mathcal{X}_r$ is bounded, and $\rho(\bm{x})$ is continuous, where $\bm{\phi}_{Z}(t)=\{\bm{x}\mid \bm{x}=\bm{\phi}_{\bm{x}_0}(t),\bm{x}_0\in Z\}$, we have that 
\[
\begin{split}
&\int_{\bm{\phi}_{Z}(t)}\rho_1(\bm{x}) d\bm{x}-\int_Z \rho_1(\bm{x}) d\bm{x}\\
&=\int_{0}^t \int_{\bm{\phi}_{Z}(\tau)}\triangledown (\rho_1 \bm{f})(\bm{x})d\bm{x}d\tau \geq 0,  \forall t\geq 0,
\end{split}
\] 
 according to $\triangledown (\rho_1 \bm{f})(\bm{x})\geq 0, \forall \bm{x}\in \overline{\mathcal{R}\setminus \mathcal{X}_r}$.  Thus,
\begin{equation}
\label{dayudengyu}
\int_{\bm{\phi}_{Z}(t)}\rho_1(\bm{x}) d\bm{x}\geq \int_Z \rho_1(\bm{x}) d\bm{x}, \forall t\geq 0.
\end{equation} 
 
 Further,  since $\rho_1(\bm{x})\leq \triangledown (\rho_2\bm{f})(\bm{x}), \forall \bm{x}\in \overline{\mathcal{X}\setminus \mathcal{X}_r}$,  we have that 
 \[\int_{0}^t \int_{\bm{\phi}_Z(\tau)} \rho_1(\bm{x})d\bm{x}\leq \int_{0}^t \int_{\bm{\phi}_{Z}(\tau)}\triangledown (\rho_2 \bm{f})(\bm{x})d\bm{x}d\tau, \forall t\geq 0. \]
 Combining \eqref{dayudengyu}, we have that 
 \[\int_{Z}\rho_1(\bm{x}) d\bm{x}\leq \frac{\int_{\bm{\phi}_{Z}(t)}\rho_2(\bm{x}) d\bm{x}-\int_Z \rho_2(\bm{x}) d\bm{x}}{t}, \forall t\geq 0\]
 and consequently, \[\int_{Z}\rho_1(\bm{x}) d\bm{x}\leq 0=\lim_{t\rightarrow \infty}\frac{\int_{\bm{\phi}_Z(t)}\rho_2(\bm{x}) d\bm{x}-\int_Z \rho_2(\bm{x}) d\bm{x}}{t}.\] 
 Since $\rho_1(\bm{x})>0$ over $Z$, we have the conclusion that $Z$ is a set of measure zero. Therefore, the set of all initial conditions in $\mathcal{R}$ whose flows stay in $\mathcal{R}\setminus \mathcal{X}_r$ for all the time is a set of measure zero.

Now take any $\bm{x}_0\in \mathcal{R}$ whose flow leaves $\mathcal{R}\setminus \mathcal{X}_r$ in finite time. We will show that such a flow must enter $\mathcal{X}_r$ before leaving $\mathcal{R}$. Suppose to the contrary that the flow $\bm{\phi}_{\bm{x}_0}(t)$ leaves $\mathcal{R}$ without entering $\mathcal{X}_r$ first. Let $T>0$ be the first time instant that $\bm{\phi}_{\bm{x}_0}(T)\in \partial \mathcal{R}$, i.e., $\rho_1(\bm{\phi}_{\bm{x}_0}(T))=0$. 

Since $\triangledown (\rho_1\bm{f})(\bm{x})\geq 0, \forall \bm{x} \in \overline{\mathcal{R}\setminus \mathcal{X}_r}$ and $\rho_1(\bm{\phi}_{\bm{x}_0}(t))\geq 0$ for $t\in [0,T]$, we have that 
\begin{equation}
\label{expoenential_in}
\begin{split}
&\frac{d \rho_1(\bm{\phi}_{\bm{x}_0}(t))}{d t}=\triangledown \rho_1(\bm{x})\cdot \bm{f}(\bm{x})\mid_{\bm{x}=\bm{\phi}_{\bm{x}_0}(t)}\\
&\geq -\rho_1(\bm{x})\triangledown \bm{f}(\bm{x})\mid_{\bm{x}=\bm{\phi}_{\bm{x}_0}(t)}\\
&\geq -\lambda_0 \rho_1(\bm{\phi}_{\bm{x}_0}(t)), \forall t\in [0,T],
\end{split}
\end{equation}
where $\lambda_0=\max_{\bm{x}\in \overline{\mathcal{X}\setminus \mathcal{X}_r}} \triangledown \bm{f}(\bm{x})$. Consequently, we have that 
\[\rho_1(\bm{\phi}_{\bm{x}_0}(T))>0,\]
contradicting $\rho_1(\bm{\phi}_{\bm{x}_0}(T))=0$. Thus, we conclude that there must  exist $t\geq 0$ such that $\bm{\phi}_{\bm{x}_0}(t)\in \mathcal{X}_r$ and $\bm{\phi}_{\bm{x}_0}(\tau)\in \mathcal{R}$ for all $\tau\in [0,t]$.

Since $\rho_1(\bm{x})>0$ for $\bm{x}\in \mathcal{X}_0$, $\mathcal{X}_0\subseteq \mathcal{R}$ holds. Consequently, the conclusion holds.  
\end{proof}

Comparing conditions \eqref{weak_exponential_con} and \eqref{weak_asympt_con}, one difference lies in that condition \eqref{weak_asympt_con} allows $\lambda(\bm{x})$ to be equal to zero over some $\bm{x}\in  \overline{\mathcal{X}\setminus \mathcal{X}_r}$.  Since the `equal' sign is taken into account, constraint \[\triangledown (\rho_1 \bm{f})(\bm{x})\geq \lambda(\bm{x})\rho_1(\bm{x}), \forall \bm{x}\in \overline{\mathcal{X}\setminus \mathcal{X}_r}\] can only ensure that all trajectories starting from $\mathcal{R}$ cannot leave the set $\mathcal{R}$ if they do not reach the target set $\mathcal{X}_r$. This conclusion can be derived from \eqref{expoenential_in}. In order to ensure the reach of the target set $\mathcal{X}_r$, a new constraint, i.e., \[\rho_1(\bm{x})\leq \triangledown (\rho_2\bm{f})(\bm{x}), \forall \bm{x}\in \overline{\mathcal{X}\setminus \mathcal{X}_r},\] is introduced. This constraint ensures that the set of initial states in $\mathcal{R}$ such that system \eqref{sys} stays inside $\mathcal{R}\setminus \mathcal{X}_r$ for all the time is a set of measure zero. That is, it ensures that system \eqref{sys} starting from almost all initial states in $\mathcal{R}$ will reach the target set $\mathcal{X}_r$ eventually while staying inside $\mathcal{R}$ before the first target hitting time.  If $\lambda(\bm{x})>0$ over $\overline{\mathcal{X}\setminus \mathcal{X}_r}$, constraint $\rho_1(\bm{x})\leq \triangledown (\rho_2\bm{f})(\bm{x}), \forall \bm{x}\in \overline{\mathcal{X}\setminus \mathcal{X}_r}$ in condition \eqref{weak_asympt_con} can be removed and thus condition \eqref{weak_asympt_con} will equal condition \eqref{weak_exponential_con}.

Also, we can show that if there exists $\rho(\bm{x})\in \mathcal{C}^1(\overline{\mathcal{X}})$ satisfying \[\triangledown (\rho \bm{f})(\bm{x})>0, \forall \bm{x}\in \overline{\mathcal{X}\setminus \mathcal{X}_r},\] there exist $\rho_1(\bm{x}),\rho_2(\bm{x})\in \mathcal{C}^1(\overline{\mathcal{X}})$ such that \[\triangledown (\rho_1 \bm{f})(\bm{x})\geq \lambda(\bm{x})\rho_1(\bm{x}), \forall \bm{x}\in \overline{\mathcal{X}\setminus \mathcal{X}_r}\] and \[\rho_1(\bm{x})\leq \triangledown (\rho_2\bm{f})(\bm{x}), \forall \bm{x}\in \overline{\mathcal{X}\setminus \mathcal{X}_r}\] hold, where $\lambda(\bm{x})\equiv 0$ for $\bm{x}\in \overline{\mathcal{X}}$. This conclusion can be certified in the following way: That \[\triangledown (\rho \bm{f})(\bm{x})>0, \forall \bm{x}\in \overline{\mathcal{X}\setminus \mathcal{X}_r}\] implies that
\[\exists \epsilon_0>0. \triangledown (\rho \bm{f})(\bm{x})\geq \epsilon_0, \forall \bm{x}\in \overline{\mathcal{X}\setminus \mathcal{X}_r}.\]  Let $M\epsilon_0 \geq \rho(\bm{x})$ for $\bm{x}\in \overline{\mathcal{X}\setminus \mathcal{X}_r}$, where $M>0$. Therefore, we can take \[\rho_1(\bm{x}):=\rho(\bm{x}), \lambda(\bm{x}):=0, \rho_2(\bm{x}):=M\rho(\bm{x})\] over $\overline{\mathcal{X}\setminus \mathcal{X}_r}$, which satisfy
\[\triangledown (\rho_1 \bm{f})(\bm{x})\geq 0, \forall \bm{x}\in \overline{\mathcal{X}\setminus \mathcal{X}_r}\] and $\rho_1(\bm{x})\leq M\epsilon_0 \leq \triangledown (\rho_2\bm{f})(\bm{x}), \forall \bm{x}\in \overline{\mathcal{X}\setminus \mathcal{X}_r}$. Therefore, condition \eqref{weak_asympt_con} is also more expressive than \eqref{weak_2}.

\begin{remark}
If $\lambda(\bm{x})$ is allowed to take negative values over $ \overline{\mathcal{X}\setminus \mathcal{X}_r}$ in Theorem \ref{sdf}, then for ensuring satisfaction of the weak reach-avoid property in Definition \ref{WRNS}, 
the constraint  $\rho_1(\bm{x})\leq \triangledown (\rho_2\bm{f})(\bm{x}), \forall \bm{x}\in \overline{\mathcal{X}\setminus \mathcal{X}_r}$ in condition \eqref{weak_2} should be $\rho_1(\bm{x})<\triangledown (\rho_2\bm{f})(\bm{x}), \forall \bm{x}\in \overline{\mathcal{X}\setminus \mathcal{X}_r}$, and the others remain the same. Due to space limitations we omit the proof here.
\end{remark}

\subsection{Generalization to Strong Reach-avoid Verification}
\label{sub:sraa}
In this subsection we exploit the differences between conditions \eqref{strong_exponen_con}/\eqref{strong_asympt_con} and \eqref{weak_exponential_con}/\eqref{weak_asympt_con}, and explore the situations, under which the sufficient conditions in Theorem \ref{exponential_almost} and \ref{sdf} can also be used to verify the strong reach-avoid property in the sense of Definition \ref{SRNS}.

 The main difference between conditions \eqref{strong_exponen_con}/\eqref{strong_asympt_con} and \eqref{weak_exponential_con}/\eqref{weak_asympt_con} lies in that condition \eqref{strong_exponen_con}/\eqref{strong_asympt_con} uses $\triangledown (\rho\bm{f})(\bm{x})$ rather than $\triangledown  \rho(\bm{x})\cdot \bm{f}(\bm{x})$. Comparing to $\triangledown \rho(\bm{x})\cdot \bm{f}(\bm{x})$, the term $\triangledown (\rho\bm{f})(\bm{x})$ has an additional term $\rho(\bm{x})\triangledown \bm{f}(\bm{x})$. Therefore, when 
 \begin{equation}
 \label{equal0}
 \triangledown \bm{f}(\bm{x})\equiv 0, \forall \bm{x} \in \overline{\mathcal{X}\setminus \mathcal{X}_r},
 \end{equation}
 we have that conditions \eqref{weak_exponential_con} and \eqref{weak_asympt_con} are respectively a special form of ones \eqref{strong_exponen_con} and \eqref{strong_asympt_con}. In this case, if condition \eqref{weak_exponential_con} or \eqref{weak_asympt_con}  holds, we can also conclude that the strong reach-avoid property in the sense of Definition \ref{SRNS} holds. We do not give the proofs here since this conclusion is just a special case of Corollary \ref{weak_strong_ex} and \ref{weak_starong_as} shown below. However, condition \eqref{equal0} may be quite restrictive in practice, limiting the use of conditions \eqref{weak_exponential_con} and \eqref{weak_asympt_con} in verifying the strong reach-avoid property. In order to overcome this issue, we in the following formulate two less conservative constraints such that the satisfaction of condition \eqref{weak_exponential_con} or \eqref{weak_asympt_con} also implies the satisfaction of the strong reach-avoid property. They are respectively formulated in  Corollary \ref{weak_strong_ex} and \ref{weak_starong_as}.

\begin{corollary}
\label{weak_strong_ex}
If there exist a density function $\rho(\bm{x})\in \mathcal{C}^1(\overline{\mathcal{X}})$ and a continuous function 
\begin{equation}
\label{expo_condition}
    \lambda(\bm{x})>\triangledown \bm{f}(\bm{x}), \forall \bm{x}\in \overline{\mathcal{X}\setminus \mathcal{X}_r},
\end{equation} 
which satisfy \eqref{weak_exponential_con}, then the strong reach-avoid property holds.
\end{corollary}
\begin{proof}
We firstly show that there does not exist an initial state $\bm{x}_0\in \mathcal{X}_0$ such that 
\[\bm{\phi}_{\bm{x}_0}(t)\in \mathcal{R}\setminus \mathcal{X}_r, \forall t\in [0,\infty),\] 
where $\mathcal{R}=\{\bm{x}\in \mathcal{X}\mid \rho(\bm{x})>0\}$.

Assume that $\bm{\phi}_{\bm{x}_0}(t)\in \mathcal{R}\setminus \mathcal{X}_r, \forall t\in [0,\infty)$ holds. From constraints $\lambda(\bm{x}) \rho(\bm{x})\leq \triangledown(\rho\bm{f})(\bm{x}), \forall \bm{x}\in \overline{\mathcal{R}\setminus \mathcal{X}_r}$ and $\lambda(\bm{x})>\triangledown \bm{f}(\bm{x}), \forall \bm{x}\in \overline{\mathcal{X}}$, we have that for $t\geq 0$,
\begin{equation}
\label{exponential}
\begin{split}
&\triangledown \rho(\bm{x})\cdot\bm{f}(\bm{x})\mid_{\bm{x}=\bm{\phi}_{\bm{x}_0}(t)}\geq \\
&(\lambda(\bm{x})-\triangledown \bm{f}(\bm{x})) \rho(\bm{x})\mid_{\bm{x}=\bm{\phi}_{\bm{x}_0}(t)}.
\end{split}
\end{equation}
Further, since $\overline{\mathcal{X}\setminus \mathcal{X}_r}$ is compact, there exists $\epsilon_0>0$ such that $\lambda(\bm{x})-\triangledown \bm{f}(\bm{x})\geq \epsilon_0, \forall \bm{x}\in \overline{\mathcal{X}\setminus \mathcal{X}_r}$. Thus, we have $\rho(\bm{\phi}_{\bm{x}_0}(t))\geq e^{\epsilon_0 t} \rho(\bm{x}_0), \forall t\in [0,\infty)$,
which contradicts that $\rho(\bm{x})$ is bounded over $\overline{\mathcal{X}}$. Therefore, these exists $\tau'\geq 0$ such that $\bm{\phi}_{\bm{x}_0}(\tau')\notin \mathcal{R}\setminus \mathcal{X}_r$. 

Besides, constraint \eqref{exponential} implies that 
\[ \rho(\bm{\phi}_{\bm{x}_0}(t))\geq e^{\epsilon_0 t}\rho(\bm{x}_0)>0, \forall t\in [0,T],\]
where $T=\max\{t\mid \forall \tau\in[0,t]. \bm{\phi}_{\bm{x}_0}(\tau) \in \mathcal{R}\setminus \mathcal{X}_r\}$.
Since $\rho(\bm{x})=0$ for $\bm{x}\in \partial \mathcal{R}$, we have that $\bm{\phi}_{\bm{x}_0}(T) \in \mathcal{X}_r$.  Since $\mathcal{R}\subseteq \mathcal{X}$, we have that the strong reach-avoid property in the sense of Definition \ref{SRNS} holds. 
\end{proof}

Corollary \ref{weak_strong_ex} indicates that when $\lambda(\bm{x})>\triangledown\bm{f}(\bm{x})$ over $\overline{\mathcal{X}\setminus \mathcal{X}_r}$, condition \eqref{weak_exponential_con} can also be used for the strong reach-avoid verification and behaves like condition \eqref{strong_exponen_con}. However, it is observed that condition \eqref{weak_exponential_con} is more expressive than condition \eqref{strong_exponen_con}, since condition \eqref{strong_exponen_con} is just a special instance of condition \eqref{weak_exponential_con} with $\lambda(\bm{x})=\triangledown \bm{f}(\bm{x})+\lambda$. \textit{Furthermore, it is interesting to find that when $\max_{\bm{x}\in \overline{\mathcal{X}\setminus \mathcal{X}_r}}\triangledown \bm{f}(\bm{x})<0$, the continuous function $\lambda(\bm{x})$ in condition \eqref{weak_exponential_con} can be further relaxed and is not necessary to be positive over $\overline{\mathcal{X}\setminus \mathcal{X}_r}$ for both the weak and strong reach-avoid verification.} In case that $\max_{\bm{x}\in \overline{\mathcal{X}\setminus \mathcal{X}_r}}\triangledown \bm{f}(\bm{x})>0$, a continuous function $\lambda(\bm{x})$ satisfying 
\begin{equation*}
\forall \bm{x}\in \overline{\mathcal{X}\setminus \mathcal{X}_r}. \lambda(\bm{x}) >0 \wedge \exists \bm{x}\in \overline{\mathcal{X}\setminus \mathcal{X}_r}. \lambda(\bm{x}) \leq \triangledown \bm{f}(\bm{x})
\end{equation*}
will render condition \eqref{weak_exponential_con} only applicable to the weak reach-avoid verification of system \eqref{sys}.

\begin{corollary}
\label{weak_starong_as}
 If there exist density functions $\rho_1(\bm{x}),\rho_2(\bm{x})\in \mathcal{C}^1(\overline{\mathcal{X}})$ and $\lambda(\bm{x}) \in \mathcal{C}(\overline{\mathcal{X}})$ satisfying condition \eqref{weak_asympt_con}, then the strong reach-avoid property in the sense of Definition \ref{SRNS} holds when  
 \[\rho_2(\bm{x})\triangledown \bm{f}(\bm{x})\leq 0 \text{~and~}\lambda(\bm{x})\geq \triangledown \bm{f}(\bm{x})\]
 for $\bm{x}\in \overline{\mathcal{X}\setminus \mathcal{X}_r}$. 
\end{corollary}
\begin{proof}
From constraints \[\rho_1(\bm{x})\leq \triangledown(\rho_2\bm{f})(\bm{x}), \forall \bm{x}\in \overline{\mathcal{R}\setminus \mathcal{X}_r}\] and $\triangledown \bm{f}(\bm{x})\rho_2(\bm{x})\leq 0, \forall \bm{x}\in \overline{\mathcal{R}\setminus \mathcal{X}_r}$, where $\mathcal{R}=\{\bm{x}\in \mathcal{X}\mid \rho(\bm{x})>0\}$, we have that 
\[\triangledown \rho_2(\bm{x})\cdot\bm{f}(\bm{x})\geq \rho_1(\bm{x}), \forall \bm{x}\in \overline{\mathcal{R}\setminus \mathcal{X}_r}.\]
Further, since $\lambda(\bm{x})\geq \triangledown \bm{f}(\bm{x})$ over $\overline{\mathcal{X}\setminus \mathcal{X}_r}$, we have that $\triangledown \rho_1(\bm{x})\cdot\bm{f}(\bm{x})\geq 0$ over $\overline{\mathcal{R}\setminus \mathcal{X}_r}$. Following the proof of  Proposition 5 in \cite{xue2022reach}, we have the conclusion.
\end{proof}

If $\lambda(\bm{x})> \triangledown \bm{f}(\bm{x})$ over $\overline{\mathcal{X}\setminus \mathcal{X}_r}$, the constraint $\rho_2(\bm{x})\triangledown \bm{f}(\bm{x})\leq 0, \forall \bm{x}\in \overline{\mathcal{X}\setminus \mathcal{X}_r}$ in Corollary \ref{weak_starong_as} is redundant since the constraint  $\rho_1(\bm{x})\leq \triangledown (\rho_2\bm{f})(\bm{x}), \forall \bm{x}\in \overline{\mathcal{X}\setminus \mathcal{X}_r}$ in condition \eqref{weak_asympt_con} can be removed, according to Corollary \ref{weak_strong_ex}.

 It is worth noting here that if  $\lambda(\bm{x}) \in \mathcal{C}(\overline{\mathcal{X}})$ in condition \eqref{weak_exponential_con} (or, \eqref{weak_asympt_con}) does not satisfy the aforementioned conditions, and it is just a continuous function over $\mathcal{C}(\overline{\mathcal{X}})$, the condition \eqref{weak_exponential_con} (or, \eqref{weak_asympt_con}) can deal with the case that the safety and performance objectives are in conflict, but the safety is prioritized. In this case system \eqref{sys} starting from $\mathcal{R}=\{\bm{x}\in \mathcal{X}\mid \rho(\bm{x})>0\}$ (or, $\mathcal{R}=\{\bm{x}\in \mathcal{X}\mid \rho_1(\bm{x})>0\}$) will stay inside the set $\mathcal{R}\setminus \mathcal{X}_r$, which is a subset of the safe set $\mathcal{X}$, if it cannot reach the target set $\mathcal{X}_r$. However, a qualitative characterization of initial states in $\mathcal{R}$ such that system \eqref{sys} enters $\mathcal{X}_r$ cannot be given.


\section{Examples}
\label{sec:exam}
In this section we demonstrate our theoretical developments on one example. The condition used for computations are relaxed into semi-definite constraints based on the sum-of-squares decomposition for multivariate polynomials. The formulated semi-definite programs are presented in Appendix. The sum-of-squares module of YALMIP \cite{lofberg2004} was used to transform the sum-of-squares optimization problem into a semi-definite program and the solver Mosek \cite{mosek2015mosek} was used to solve the resulting semi-definite program.

\begin{example}
\label{ex1}
Consider an academic example from \cite{xue2022reach},
\begin{equation}
    \begin{cases}
           &\dot{x}=-0.5x-0.5y+0.5xy\\
           &\dot{y}=-0.5y+0.5
    \end{cases}
\end{equation}
with $\mathcal{X}=\{(x,y)^{\top}\mid x^2 + y^2 - 1 <0\}$, $\mathcal{X}_r=\{(x,y)^{\top}\mid (x+0.2)^2 + (y -0.7)^2 - 0.02\leq 0\}$ and $\mathcal{X}_0=\{(x,y)^{\top}\mid (x - 0.3)^2 + (y + 0.6)^2 -0.01<0\}$.

In this experiment we take $\lambda(\bm{x})\equiv$  Constant over $\overline{\mathcal{X}}$ in conditions \eqref{weak_exponential_con} and \eqref{weak_asympt_con}. 

Due to the presence of multiple unknown polynomials in solving semi-definite programs \eqref{sos2}, \eqref{sos1}, \eqref{sos3}, \eqref{sos4} and \eqref{sos5}, we use the following procedure for automatically assigning parametric templates to these polynomials.  Given degree $d$, the used polynomial templates are ones including all monomials of degree less than or equal to $d$. In the following procedure, $d_{\rho}$ and $d_s$ respectively denote the degree of the polynomials $\{\rho_1(\bm{x}),\rho_2(\bm{x}),\rho(\bm{x}),v(\bm{x}),w(\bm{x})\}$ and $\{s_i(\bm{x}),i=0,1,2,3,4,p(\bm{x})\}$. The degrees of polynomials used for verifying strong/weak properties successfully via solving these SDPs are presented in Table \ref{tab-1}. Some of computed $\mathcal{R}$'s are visualized in Fig. \ref{fig_1}. Since $\max_{\bm{x}\in \overline{\mathcal{X}\setminus \mathcal{X}_r}} \triangledown \bm{f}(\bm{x})\leq -0.50$, $\lambda<0$ is also allowed in conditions \eqref{weak_exponential_con} and \eqref{weak_asympt_con} for performing verification, and that condition \eqref{weak_exponential_con} or \eqref{weak_asympt_con} holds  also implies the satisfaction of the strong reach-avoid property according to Corollary  \ref{weak_strong_ex} and \ref{weak_starong_as}. 

Besides, it is interesting to find from Table \ref{tab-1} that conditions \eqref{weak_exponential_con} and \eqref{weak_asympt_con} are also able to facilitate the weak/strong reach-avoid verification efficiently for some cases.

\begin{algorithm}
    \begin{algorithmic}
        \FOR{$d_{\rho}=6:1:12$}
           \FOR{$d_s= 2\lceil \frac{d_{\rho}}{2} \rceil :2:2d_{\rho}$}
              \STATE solve  \eqref{sos2}  $\setminus$ \eqref{sos1} $\setminus$ \eqref{sos3} $\setminus$ \eqref{sos4} $\setminus$ \eqref{sos5}
              \IF{Solved Successfully}
                \RETURN $d_{\rho}$, $d_s$ and $\rho(\bm{x})$
              \ENDIF
           \ENDFOR
        \ENDFOR
    \end{algorithmic}
\end{algorithm}

\begin{table}
\begin{center}
\begin{tabular}{|c|c|c|c|c|}
  \hline
    SDP& $\lambda$&$d_{\rho}$&$d_{s}$\\\hline
     \eqref{sos2} &0.001 &6  &12 \\\hline
     \eqref{sos2} &-0.499 &6  &6 \\\hline
     \eqref{sos1} & 0 & 6 &6  \\\hline
     \eqref{sos3} & -    &6  &12\\\hline
     \eqref{sos4} &-& 10 &10 \\\hline
     \eqref{sos5} &0.001 & 10 &10 \\\hline
   \end{tabular}
\end{center}~
\caption{Parameters of solving SDP \eqref{sos2}-\eqref{sos5} to verify the strong/weak reach-avoid properties successfully ('-' means that $\lambda$ is not used).}
\label{tab-1}
\end{table}

\begin{figure}[hbtp]
\center
\subfigure[]{\includegraphics[width=0.20\textwidth]{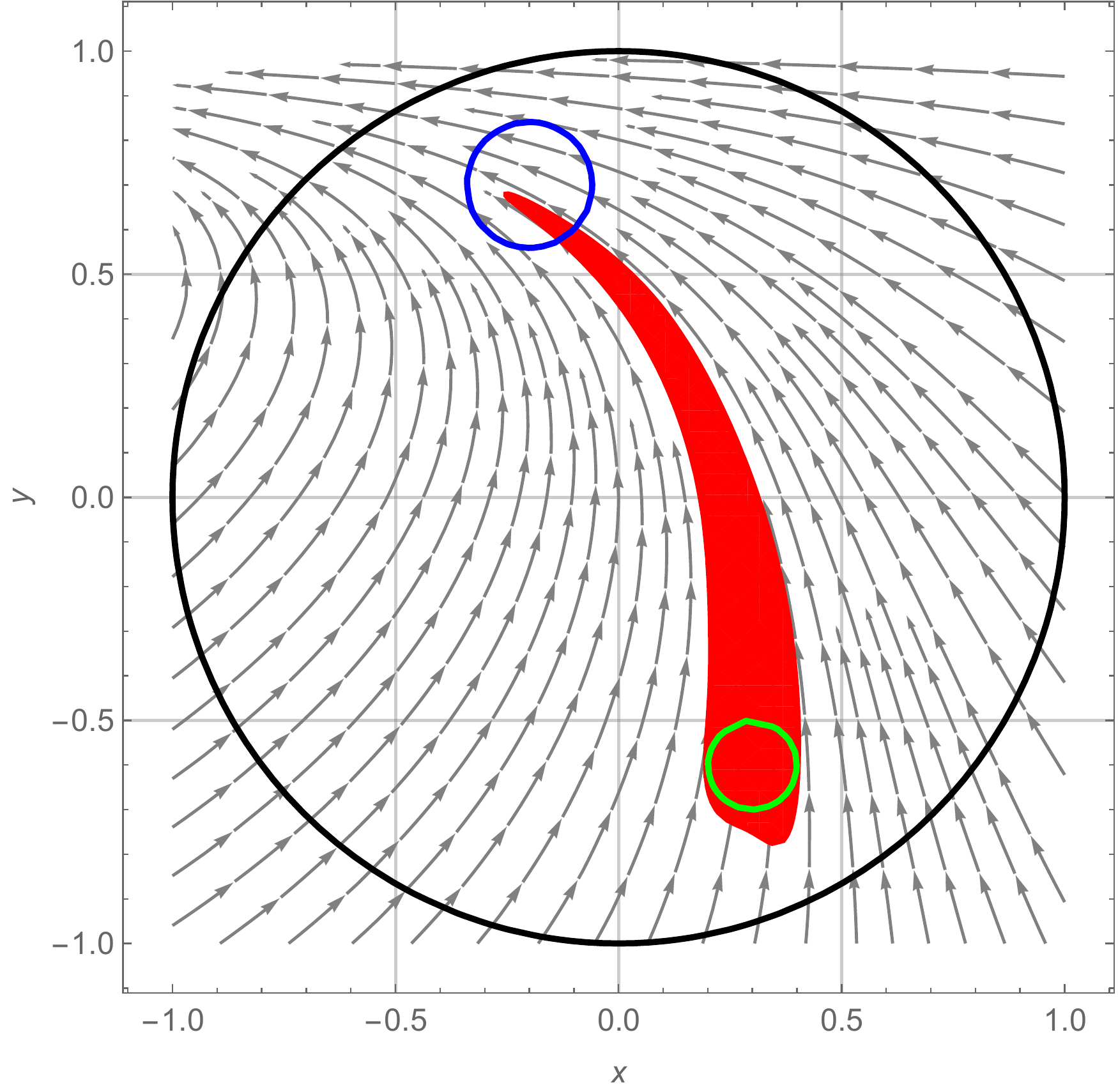}} 
\subfigure[]{\includegraphics[width=0.20\textwidth]{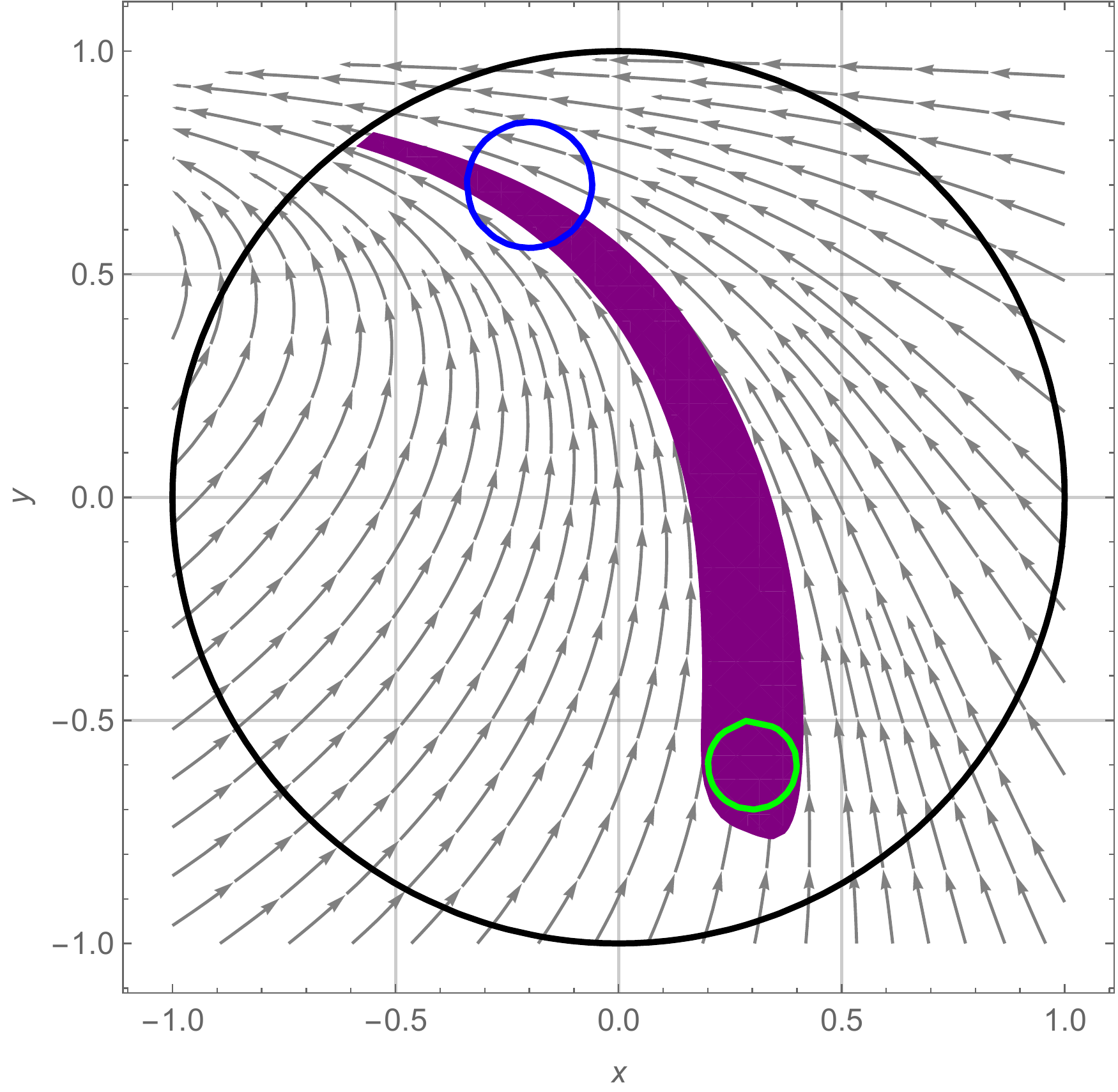}}
\caption{Blue, black and green curves - $\partial \mathcal{X}_r$, $\partial \mathcal{X}$ and $\partial \mathcal{X}_0$; red region - $\{\bm{x}\in \mathcal{X}\mid \rho_1(\bm{x})>0\}$ via solving SDP \eqref{sos1}; purple region - $\{\bm{x}\in \mathcal{X}\mid \rho(\bm{x})>0\}$ via solving SDP \eqref{sos2} with $\lambda=-0.499$.}
\label{fig_1}
\end{figure}

\end{example}

\section{Conclusion}
\label{sec:conclusion}
In this paper we investigated the reach-avoid verification of continuous-time systems modeled by ordinary differential equations using Lyapunov densities. Two new sufficient conditions were proposed for the weak reach-avoid verification, which are shown to be weaker than existing ones. Then, via analyzing the divergence of the vector field and constraining it, we generalized the proposed two conditions to the strong reach-avoid verification.  Finally, we demonstrated our theoretical developments on one example. The experimental results also showed that the proposed conditions can facilitate the weak/strong reach-avoid verification efficiently.

An appealing point of Lyapunov densities lies in facilitating the control design using convex optimization, especially for control-affine dynamics \cite{1266794}. In the future we would investigate the reach-avoid controller synthesis based on the proposed conditions in the present work.  
\bibliographystyle{abbrv}
\bibliography{reference}
\section*{Appendix}
\footnotesize
\begin{algorithm}
 The semi-definite program for solving constraint \eqref{weak_exponential_con}:
\begin{equation}
\label{sos2}
 \begin{cases}
 \triangledown (\rho \bm{f})(\bm{x})-\lambda(\bm{x})\rho(\bm{x})+s_0(\bm{x}) h(\bm{x})-s_1(\bm{x})g(\bm{x})\in \sum[\bm{x}],\\
 \rho(\bm{x})-\epsilon_0+s_2(\bm{x})l(\bm{x})\in \sum[\bm{x}],\\
 -\rho(\bm{x})+p(\bm{x})h(\bm{x})\in \sum[\bm{x}],
 \end{cases}
 \end{equation}
 where $\epsilon_0=10^{-6}$, $\rho(\bm{x}), p(\bm{x}), \bm{\psi}_i(\bm{x})\in \mathbb{R}[\bm{x}]$, and $s_{j}(\bm{x})\in \sum[\bm{x}]$, $i=0,\ldots,2$.
\end{algorithm}
\vspace{-1cm}
\begin{algorithm}
The semi-definite program for solving constraint \eqref{weak_asympt_con}:
\begin{equation}
\label{sos1}
 \begin{cases}
 \triangledown (\rho_1 \bm{f})(\bm{x})-\lambda(\bm{x})\rho_1(\bm{x})+s_0(\bm{x}) h(\bm{x})-s_1(\bm{x})g(\bm{x})\in \sum[\bm{x}],\\
 \triangledown (\rho_2\bm{f})(\bm{x})-\rho_1(\bm{x})+s_2(\bm{x}) h(\bm{x})-s_3(\bm{x})g(\bm{x})\in \sum[\bm{x}],\\
 \rho_1(\bm{x})-\epsilon_0+s_4(\bm{x})l(\bm{x})\in \sum[\bm{x}],\\
 -\rho_1(\bm{x})+p(\bm{x})h(\bm{x})\in \sum[\bm{x}],
 \end{cases}
 \end{equation}
 where $\epsilon_0=10^{-6}$, $\rho(\bm{x}), p(\bm{x}), \bm{\psi}_i(\bm{x})\in \mathbb{R}[\bm{x}]$, and $s_{j}(\bm{x})\in \sum[\bm{x}]$, $i=0,\ldots,4$.
 \end{algorithm}
 \vspace{-1cm}
\begin{algorithm}
  The semi-definite program for solving constraint \eqref{weak_2}
\begin{equation}
\label{sos3}
 \begin{cases}
        \triangledown(\rho\bm{f})(\bm{x})-\epsilon'_0+s_0(\bm{x}) h(\bm{x})-s_1(\bm{x})g(\bm{x})\in \sum[\bm{x}],\\
         \rho(\bm{x})-\epsilon_0+s_2(\bm{x})l(\bm{x})\in \sum[\bm{x}],\\
         -\rho(\bm{x})+p(\bm{x})h(\bm{x})\in \sum[\bm{x}],
\end{cases}
 \end{equation}
 where $\epsilon_0=10^{-6}$, $\rho(\bm{x}), p(\bm{x}), \bm{\psi}_i(\bm{x})\in \mathbb{R}[\bm{x}]$, and $s_{j}(\bm{x})\in \sum[\bm{x}]$, $i=0,\ldots,2$.
 \end{algorithm}
 \vspace{-0.5cm}
\begin{algorithm}
The semi-definite program for solving constraint \eqref{strong_asympt_con}:
 \begin{equation}
 \label{sos4}
\begin{cases}
\bigtriangledown  v(\bm{x}) \cdot \bm{f}(\bm{x})+s_0(\bm{x}) h(\bm{x})-s_1(\bm{x})g(\bm{x})\in \sum[\bm{x}],\\
\bigtriangledown  w(\bm{x}) \cdot \bm{f}(\bm{x})-v(\bm{x})+s_2(\bm{x}) h(\bm{x})-s_3(\bm{x})g(\bm{x})\in \sum[\bm{x}],\\
v(\bm{x})-\epsilon_0+s_4(\bm{x})l(\bm{x})\in \sum[\bm{x}],\\
-v(\bm{x})+p(\bm{x})h(\bm{x})\in \sum[\bm{x}],
\end{cases}
\end{equation}
 where $\epsilon_0=10^{-6}$, $v(\bm{x}), w(\bm{x}), p(\bm{x})\in \mathbb{R}[\bm{x}]$, and $s_{j}(\bm{x})\in \sum[\bm{x}]$, $i=0,\ldots,4$.
\end{algorithm}
\vspace{-0.6cm}
\begin{algorithm}
The semi-definite program for solving constraint \eqref{strong_exponen_con}:
\begin{equation}
\label{sos5}
 \begin{cases}
 \bigtriangledown  v(\bm{x}) \cdot \bm{f}(\bm{x})-\lambda v(\bm{x})+s_0(\bm{x}) h(\bm{x})-s_1(\bm{x})g(\bm{x})\in \sum[\bm{x}],\\
 v(\bm{x})-\epsilon_0+s_2(\bm{x})l(\bm{x})\in \sum[\bm{x}],\\
 -v(\bm{x})+p(\bm{x})h(\bm{x})\in \sum[\bm{x}],
 \end{cases}
 \end{equation}
 where $\epsilon_0=10^{-6}$, $v(\bm{x}), p(\bm{x})\in \mathbb{R}[\bm{x}]$, and $s_{j}(\bm{x})\in \sum[\bm{x}]$, $i=0,\ldots,2$.
\end{algorithm}
\vspace{-0.5cm}

\end{document}